\pdfoutput=1
\documentclass[11pt]{article}
\usepackage[margin=1in]{geometry}
\usepackage[utf8]{inputenc}
\usepackage[T1]{fontenc}
\usepackage{lmodern}
\usepackage{amsmath,amssymb,amsthm}
\usepackage{aliascnt}
\usepackage{booktabs}
\usepackage{graphicx}
\usepackage{tikz}
\usepackage{enumitem}
\usepackage{xcolor}
\usepackage[hidelinks]{hyperref}
\usepackage[capitalise]{cleveref}

\newtheorem{theorem}{Theorem}
\newaliascnt{lemma}{theorem}
\newtheorem{lemma}[lemma]{Lemma}
\aliascntresetthe{lemma}
\newaliascnt{proposition}{theorem}
\newtheorem{proposition}[proposition]{Proposition}
\aliascntresetthe{proposition}
\newaliascnt{corollary}{theorem}
\newtheorem{corollary}[corollary]{Corollary}
\aliascntresetthe{corollary}
\theoremstyle{definition}
\newaliascnt{definition}{theorem}
\newtheorem{definition}[definition]{Definition}
\aliascntresetthe{definition}
\theoremstyle{remark}
\newaliascnt{remark}{theorem}
\newtheorem{remark}[remark]{Remark}
\aliascntresetthe{remark}
\crefname{lemma}{Lemma}{Lemmas}
\Crefname{lemma}{Lemma}{Lemmas}
\crefname{proposition}{Proposition}{Propositions}
\Crefname{proposition}{Proposition}{Propositions}
\crefname{corollary}{Corollary}{Corollaries}
\Crefname{corollary}{Corollary}{Corollaries}
\crefname{definition}{Definition}{Definitions}
\Crefname{definition}{Definition}{Definitions}
\crefname{remark}{Remark}{Remarks}
\Crefname{remark}{Remark}{Remarks}

\newcommand{\F}{\mathbb{F}}
\newcommand{\Qp}{\mathcal{Q}^{+}}
\newcommand{\Qm}{\mathcal{Q}^{-}}
\newcommand{\dg}{\operatorname{deg}}
\newcommand{\Sp}{\mathrm{Sp}}
\newcommand{\GO}{\mathrm{O}}
\newcommand{\obs}[1]{\texttt{#1}}

\newcommand{\repo}{\url{https://github.com/BrianLi009/four-qubit-hyperbolic-quadric}}

\title{The Four-Qubit Hyperbolic Quadric Has Contextuality Degree 315}
\author{
  Zhengyu Li\\
  Georgia Institute of Technology\\
  \texttt{brian.li@gatech.edu}
  \and
  Amandeep Gupta\\
  Georgia Institute of Technology\\
  \texttt{alnu46@gatech.edu}
  \and
  Vijay Ganesh\\
  Georgia Institute of Technology\\
  \texttt{vganesh45@gatech.edu}
}
\date{}

\begin{document}
\maketitle

\begin{abstract}
The contextuality degree of a set of Pauli contexts is the least number of contexts that every noncontextual $\pm1$ assignment violates.
For the four-qubit hyperbolic quadric $\Qp(7,2)$, with $135$ points and $1575$ lines as contexts, Muller \emph{et al.}\ conjectured from heuristic searches that the degree is $315$ and that the violated contexts of every optimal assignment form a dual polar space $DW(5,2)$.
We prove both conjectures.
The lower bound has a short proof by hand: every doily ($GQ(2,2)$) of the quadric contains ten Peres--Mermin grids covering each of its lines four times, so every assignment violates at least three of its fifteen lines, and averaging over the doilies gives $1575/5=315$.
The equality case is local: an assignment is optimal exactly when it violates one ruling of one grid in every doily.
Using this as a SAT encoding, we show that there are exactly $240$ optimal violated sets, one orbit of $\GO^+(8,2)$, each isomorphic to $DW(5,2)$.
The same argument gives the first lower bounds $\lceil L/5\rceil$ for all quadrics with $L$ lines in $N\ge4$ qubits.
The degree theorem is kernel-checked in Lean~4, and the completeness of the enumeration is certified by a DRAT proof checked by the verified checker \texttt{cake\_lpr}.
\end{abstract}

\section{Introduction}
The Kochen--Specker theorem~\cite{KS67,Bell66} rules out noncontextual hidden-variable models of quantum mechanics: there is no assignment of definite values to observables, independent of which compatible observables are measured alongside, that respects the functional relations quantum mechanics imposes on every set of commuting observables.
The most transparent proofs use Pauli observables.
In the Peres--Mermin square~\cite{Peres90,Mermin90,Mermin93} nine two-qubit Pauli observables are arranged in a $3\times3$ grid whose rows and columns are triples of commuting observables with product $+I$, except for one column with product $-I$. No assignment of $\pm1$ values satisfies all six constraints.
Such state-independent proofs underlie experimental tests of contextuality~\cite{Cabello08,Kirchmair09,Amselem09} and connect contextuality to the power of quantum computation~\cite{Raussendorf13,Howard14}.

The non-identity $N$-qubit Pauli observables, taken up to sign, are the points of the symplectic polar space $W(2N-1,2)$ of rank $N$ over $\F_2$, two observables commuting exactly when the corresponding points are orthogonal~\cite{SanigaPlanat07,Havlicek09,Thas09}.
This dictionary has made finite geometry a natural language for observable-based contextuality~\cite{Holweck17,Trandafir22,Levay13,SanigaLevayPracna12}.
Any family of totally isotropic lines, each line $\{A,B,AB\}$ carrying the sign of $A\cdot B\cdot AB=\pm I$, defines a candidate contextuality configuration.
De Boutray, Holweck, Giorgetti, Masson and Saniga~\cite{deBoutray22} introduced the \emph{contextuality degree}, the least number of contexts that a noncontextual assignment must violate, as a measure of how contextual such a configuration is.
It governs both the classical bound of the associated noncontextuality inequality~\cite{Cabello10,Kelleher24} and the classical winning probability of the associated Mermin-type games~\cite{Arkhipov12,Kelleher24}.
Computing it is a MAX-XOR-SAT problem, and exact values are known only for small configurations: the doily ($3$), the three-qubit space $W(5,2)$ ($63$) and its quadrics ($21$ and $9$)~\cite{deBoutray22,Muller22doilies,Saniga24hex}.

For four qubits the problem has resisted computation.
The hyperbolic quadric $\Qp(7,2)$ has $135$ points and $1575$ lines.
Muller, Saniga, Giorgetti, de Boutray and Holweck~\cite[Table~1]{Muller23} could only show that its degree is at most $500$.
Muller, Saniga, Giorgetti, Holweck and Kelleher~\cite[\S5.1.2]{Muller24} then found, with a heuristic search, assignments violating only $315$ lines, observed that the $315$ unsatisfied lines are isomorphic to the dual polar space $DW(5,2)$, and wrote: ``So, we do believe that $d^{hyp}_4=315$, with the understanding that the corresponding unsatisfied configurations are (always) isomorphic to $DW(5,2)$.''
Saniga, Holweck, Kelleher, Muller, Giorgetti and de Boutray~\cite[p.~24]{Saniga24hex} single out a closely related question as ``computationally rather challenging'', whose affirmative answer ``would mean, among other things, that the degree of contextuality of a four-qubit hyperbolic quadric is 315''.
Muller and Giorgetti~\cite[\S5.3]{MullerGiorgetti24} and Muller's thesis~\cite{MullerThesis} list $315$ as an upper bound only.

\paragraph{Results.}
We settle both parts of this conjecture.
\begin{itemize}[leftmargin=*]
\item \textbf{\Cref{thm:main}.} Every hyperbolic quadric of the four-qubit space, with its $1575$ lines as contexts, has contextuality degree exactly $315$.
\item \textbf{\Cref{prop:ruling,thm:dw}.} In every optimal assignment, each of the $40\,320$ doilies of the quadric has exactly three violated lines, which form one ruling of one of its grids. There are exactly $240$ sets of lines that arise as the violated set of an optimal assignment; they form one orbit of the orthogonal group $\GO^+(8,2)$, and each is isomorphic to $DW(5,2)$. The first statement has a short proof by hand; the second is computer-assisted.
\item \textbf{\Cref{thm:general}.} For every $N\ge3$, every hyperbolic or elliptic quadric of the $N$-qubit space, with $L$ lines, has degree at least $\lceil L/5\rceil$. This is tight for $N=3$ and for the four-qubit hyperbolic quadric, and to our knowledge it gives the first lower bounds for the elliptic four-qubit quadric and for all quadrics with $N\ge5$ qubits.
\item \textbf{\Cref{cor:inequality}.} The $1575$ contexts give a state-independent noncontextuality inequality with classical bound exactly $945$ and quantum value $1575$; equivalently, a uniformly random context is satisfied with probability at most $4/5$ by any noncontextual assignment, and this bound is attained.
\end{itemize}

The proof of the lower bound is short and needs no computation, and we regard this as its main point: a value that resisted solver-based attacks follows from three classical ingredients.
These are the parity argument behind the Peres--Mermin square; Sylvester's model of the generalised quadrangle $GQ(2,2)$, which shows that a doily contains ten grids covering each of its lines four times; and Witt's theorem, which makes the orthogonal group transitive on the lines of the quadric.
The averaging principle itself is not new. Saniga \emph{et al.}~\cite[\S4.1]{Saniga24hex} obtained the degree $63$ of $W(5,2)$ by averaging over its doilies, and Muller \emph{et al.}~\cite[\S4]{Muller24} bounded the degree of the whole space $W(2N-1,2)$ by a recursion over subspaces; as far as we know, it had not been applied to the quadrics, where it turns out to be tight.
The rounding $\lceil 10/4\rceil=3$ inside each doily is what makes the bound tight: averaging over the $67\,200$ grids of the quadric directly gives only $\lceil1575/6\rceil=263$.
What we regard as the main new content is what the equality case gives: a local characterisation of all optimal assignments (\cref{prop:ruling}), which turns the classification of the optima into an easy and certified computation (\cref{thm:dw}), together with a formal proof of the degree.

We formalised the theorem for the real quadric (the observables with an even number of $Y$ factors) in Lean~4 with Mathlib~\cite{Lean4,Mathlib}; the kernel checks every finite fact, and the statement depends only on the standard axioms (\cref{sec:lean}).
All code and certificates are available at \repo.

\paragraph{For the formal-methods reader.}
Computing a contextuality degree is a MAX-XOR-SAT problem, here with $135$ variables and $1575$ parity constraints of length three, and this one had resisted the solver-based attacks of~\cite{Muller23,Muller24}.
SAT solving has settled several open problems in combinatorics, often with the help of mathematical insight in the encoding~\cite{Heule16,Zulkoski15}.
The present case is an extreme instance of that pattern: the mathematics settles the optimisation problem completely, and what remains for the solver, the classification of all optima, becomes easy only after the proof is turned into constraints.
\Cref{sec:sat,sec:ablation} give the encoding and a comparison: direct MaxSAT and cardinality formulations do not finish within $30$ minutes, while local cardinality constraints over doilies, together with redundant clauses derived from \cref{prop:doily}, enumerate all optimal solutions in about $30$ seconds; the final unsatisfiability, which gives completeness, is certified by a DRAT proof checked by \texttt{drat-trim} and by the verified checker \texttt{cake\_lpr}.
\Cref{sec:lean} describes the Lean proof, in particular how $191$ kernel evaluations were restructured to fit in $2.2$\,GB of memory, and \cref{sec:costs} lists the running time of every computation.

\paragraph{Outline.}
\Cref{sec:related} reviews related work and \cref{sec:prelim} sets up the geometry and the symmetries of the problem.
\Cref{sec:doily} proves that every doily of a quadric has degree at least $3$.
\Cref{sec:main} proves the main theorem, \cref{sec:optimal} studies the optimal assignments, and \cref{sec:general} extends the lower bound to all quadrics.
\Cref{sec:physics} discusses the physical consequences, \cref{sec:lean} the Lean proof, \cref{sec:computations} the encodings, certificates and running times, and \cref{sec:open} open problems.

\section{Related work}\label{sec:related}

\paragraph{Contextuality.}
The Kochen--Specker theorem~\cite{KS67,Bell66} and its many proofs are surveyed by Budroni, Cabello, G\"uhne, Kleinmann and Larsson~\cite{Budroni22}.
Proofs with Pauli observables began with the Peres--Mermin square and Mermin's pentagram~\cite{Peres90,Mermin90,Mermin93}; Waegell and Aravind constructed families of observable-based proofs in the $N$-qubit Pauli group~\cite{Waegell13}, and Trandafir, Lison\v{e}k and Cabello characterised the irreducible magic sets of $n$ qubits~\cite{Trandafir22}.
State-independent contextuality was turned into experimentally testable inequalities by Cabello~\cite{Cabello08,Cabello10} and tested with trapped ions and photons~\cite{Kirchmair09,Amselem09}; tests of whole symplectic polar spaces on noisy quantum computers were reported by Holweck~\cite{Holweck21}.
Quantitative measures of contextuality include the contextual fraction of Abramsky, Barbosa and Mansfield~\cite{Abramsky17}, in the sheaf-theoretic framework of~\cite{Abramsky11}; the contextuality degree studied here is a combinatorial measure for observable-based proofs, closely tied to the classical value of Mermin-type games~\cite{Arkhipov12,Kelleher24}.
Contextuality is a resource for quantum computation: for magic-state injection~\cite{Howard14}, for measurement-based computation~\cite{Raussendorf13}, and for qubit computation models~\cite{BermejoVega17}.

\paragraph{The finite geometry of Pauli observables.}
The identification of the $N$-qubit Pauli group modulo phases with the symplectic polar space $W(2N-1,2)$ is due to Saniga and Planat~\cite{SanigaPlanat07}, with general qudit versions by Havlicek, Odehnal and Saniga~\cite{Havlicek09} and Thas~\cite{Thas09}.
This dictionary has been used to classify contextual configurations: Mermin pentagrams as ovoids and their three-qubit taxonomy~\cite{SanigaLevay12,LevaySzabo17}, the Veldkamp spaces of multiple qubits~\cite{VranaLevay10}, the split Cayley hexagon of order two and its role for three qubits~\cite{LevaySanigaVrana08,Planat13,Holweck22hex}, the polar subspaces of small rank~\cite{Saniga21}, and the real four-qubit Pauli group as a hyperbolic quadric of $PG(7,2)$ with its ovoids~\cite{SanigaLevayPracna12}.
The Lagrangian Grassmannian $LGr(N,2N)$ relates $N$-qubit and $2^{N-1}$-qubit Pauli groups~\cite{Holweck14,Levay13}; for $N=3$ it maps the $135$ planes of $W(5,2)$ to the points of the four-qubit hyperbolic quadric, which is how $DW(5,2)$ appears in our problem.
For the geometry itself we rely on~\cite{PayneThas09,HirschfeldThas16,Taylor92,DeBruyn06}.

\paragraph{The contextuality degree.}
De Boutray, Holweck, Giorgetti, Masson and Saniga~\cite{deBoutray22} introduced the degree and computed it for the quadrics of three qubits ($21$ and $9$).
Muller \emph{et al.}\ enumerated the multi-qubit doilies and showed that every doily has degree $3$~\cite{Muller22doilies}; they then computed degrees of many configurations with SAT solvers and sparse linear algebra, found that the whole three-qubit space $W(5,2)$ has degree $63$, and gave upper bounds for four and five qubits~\cite{Muller23}.
A heuristic search then gave the upper bounds of \cref{tab:general} for four to seven qubits, together with the conjectures that we prove~\cite{Muller24}.
Saniga \emph{et al.}\ explained the value $63$ by averaging over doilies and showed that the unsatisfied configurations of $W(5,2)$ are classically embedded split Cayley hexagons~\cite{Saniga24hex,Holweck22hex}; Muller and Giorgetti identified an abstract structure that determines the degree of observable-based proofs~\cite{MullerGiorgetti24}, and Muller's thesis collects the algorithmic side~\cite{MullerThesis}.
Our lower bound uses the same averaging principle as~\cite{Saniga24hex,Muller24}; what is new is that, applied to the doilies of a quadric, it is tight, and that its equality case characterises the optimal assignments.

\paragraph{SAT solving and formal proof in mathematics.}
SAT solvers have settled several long-open combinatorial problems, among them the Boolean Pythagorean triples problem~\cite{Heule16}, Schur number five~\cite{Heule18} and the last open dimension of Keller's conjecture~\cite{Brakensiek20}.
Combining SAT solvers with computer algebra (SAT+CAS) was proposed in MathCheck~\cite{Zulkoski15}.
For Kochen--Specker sets in three dimensions, the minimum size was bounded from below by exhaustive search~\cite{Uijlen16}, then by SAT modulo symmetries with co-certificate learning~\cite{Kirchweger23} and by SAT+CAS with certified proofs~\cite{LiBrightGanesh24}.
Unsatisfiability results are made checkable by clausal proofs: DRAT proofs checked by \texttt{drat-trim}~\cite{DRATtrim}, and verified checkers such as \texttt{cake\_lpr}~\cite{cakeLPR}.
Parity constraints such as ours are handled natively by solvers with Gaussian elimination~\cite{Soos09}; MaxSAT solvers such as RC2~\cite{RC2} address the optimisation version directly.
Large finite computations have been formally verified before in proof assistants, notably in the proofs of the four colour theorem in Coq~\cite{Gonthier08} and of the Kepler conjecture in HOL Light and Isabelle~\cite{Hales17}.
Our use of SAT differs from these in that the solver is not used to establish the main theorem: the degree is proved by hand and kernel-checked in Lean~\cite{Lean4,Mathlib}, and the solver is used only for the classification of the optima, where the encoding is made tractable by lemmas from the proof (\cref{sec:ablation}).

\section{Pauli observables, quadrics and the contextuality degree}\label{sec:prelim}

\subsection{The symplectic polar space}
Write $\sigma_{00}=I$, $\sigma_{10}=X$, $\sigma_{01}=Z$, $\sigma_{11}=Y$ for the Pauli matrices.
A vector $v=(x\mid z)\in V=\F_2^{2N}$ labels the Hermitian $N$-qubit Pauli observable
\[P(v)=\sigma_{x_1z_1}\otimes\cdots\otimes\sigma_{x_Nz_N},\]
and every Pauli observable is $\pm P(v)$ for a unique $v$.
For $u,v\in V$ the products satisfy $P(u)P(v)=i^{\varphi(u,v)}P(u+v)$ for an exponent $\varphi(u,v)\in\mathbb Z/4$, which is the sum over the qubits of the single-qubit exponents ($XY=iZ$, $YZ=iX$, $ZX=iY$, and $-i$ for the reversed products).
Consequently $P(u)$ and $P(v)$ commute if and only if the symplectic form
\[B(u,v)=\sum_{j=1}^N\bigl(x_j(u)z_j(v)+z_j(u)x_j(v)\bigr)\]
vanishes, and they anticommute otherwise.
The points of $W(2N-1,2)$ are the nonzero vectors of $V$, that is the $4^N-1$ non-identity observables up to sign, and its lines are the triples $\{u,v,u+v\}$ with $B(u,v)=0$, that is the triples of pairwise commuting observables whose product is $\pm I$.

\subsection{Quadrics and the real observables}
The quadratic forms on $V$ with polar form $B$ are
\[Q_a(v)=Q_0(v)+B(a,v)\quad(a\in V),\qquad Q_0(v)=\sum_{j=1}^N x_j(v)z_j(v),\]
where $Q_0(v)$ counts the $Y$ factors of $P(v)$ modulo $2$.
Since $P(v)^{\mathsf T}=(-1)^{Q_0(v)}P(v)$, the singular points of $Q_0$ are exactly the \emph{real} (symmetric) Pauli observables~\cite{SanigaLevayPracna12}.
The form $Q_0$ is hyperbolic, an orthogonal sum of $N$ hyperbolic planes, and $Q_a(v)=Q_0(v+a)+Q_0(a)$, so $Q_a$ is hyperbolic when $Q_0(a)=0$ and elliptic when $Q_0(a)=1$.
For four qubits there are $136$ hyperbolic quadrics $\Qp(7,2)$ and $120$ elliptic quadrics $\Qm(7,2)$.
A hyperbolic quadric of $W(7,2)$ has $135$ points and $1575$ lines, with $35$ lines through each point and $3$ points on each line~\cite{HirschfeldThas16}.
The lines of a quadric are the lines of $W(7,2)$ all of whose points are singular; two singular points $p\ne q$ are collinear on the quadric if and only if $B(p,q)=0$, since $Q(p+q)=Q(p)+Q(q)+B(p,q)$.

\subsection{Contexts, signs and the degree}
\begin{definition}
A \emph{configuration} $\mathcal C$ is a set of points of $W(2N-1,2)$ together with a set of lines on them, the \emph{contexts}.
The \emph{sign} of a context $C=\{u,v,u+v\}$ is the bit $s(C)$ with $P(u)P(v)P(u+v)=(-1)^{s(C)}I$.
An \emph{assignment} is a map $x$ from the points to $\F_2$, the value $x(p)$ standing for the outcome $(-1)^{x(p)}$ of $P(p)$.
It \emph{violates} $C$ if $x(u)+x(v)+x(u+v)\neq s(C)$, and the \emph{contextuality degree} $\dg\mathcal C$ is the least number of contexts violated by an assignment~\cite{deBoutray22}.
\end{definition}

A noncontextual hidden-variable model assigns to each observable a fixed outcome, and quantum mechanics requires the outcomes of $P(u),P(v),P(u+v)$ to multiply to $(-1)^{s(C)}$; so $\dg\mathcal C$ is the least number of these constraints that any such model must break.
For the configuration $\mathcal H_a$ of all points and lines of the quadric $Q_a$ we write $\dg\mathcal H_a$.

Two kinds of symmetry leave the degree unchanged.

\begin{lemma}[Re-signing]\label{lem:resign}
Replacing the observable $P(p)$ by $-P(p)$ toggles the sign of every context through $p$. The assignment $x+\delta_p$ (which changes the value at $p$ only) violates, after the replacement, exactly the contexts that $x$ violates before it. Hence the replacement preserves the degree.
\end{lemma}

\begin{lemma}[Clifford transport]\label{lem:clifford}
Let $U$ be a Clifford unitary, with $UP(v)U^\dagger=(-1)^{\tau(v)}P(gv)$ for a symplectic map $g$ of $V$ and a function $\tau$.
If $C$ is a context of $\mathcal C$, then $gC$ has sign $s(gC)=s(C)+\tau(u)+\tau(v)+\tau(u+v)$.
Consequently $y(gv)=x(v)+\tau(v)$ violates $gC$ exactly when $x$ violates $C$, and $\dg g\mathcal C=\dg\mathcal C$.
\end{lemma}
\begin{proof}
Conjugate $P(u)P(v)P(u+v)=(-1)^{s(C)}I$ by $U$.
\end{proof}

Every symplectic map of $V$ is induced by a Clifford unitary~\cite{Gottesman98,DehaeneDeMoor03}.
For instance, for $r\in V$ the unitary $(I+iP(r))/\sqrt2$ induces the symplectic transvection $t_r(v)=v+B(v,r)\,r$: it fixes $P(v)$ when $B(v,r)=0$ and sends it to $iP(r)P(v)=\pm P(v+r)$ otherwise.
If $Q_0(r)=0$ then $Q_0(t_rv)=Q_0(v)+B(v,r)=Q_r(v)$, so $t_r$ carries the quadric $Q_r$ onto $Q_0$.
Together with \cref{lem:clifford}, this gives:

\begin{corollary}\label{cor:allquadrics}
All $136$ hyperbolic quadrics of $W(7,2)$ have the same contextuality degree, and so do all $120$ elliptic quadrics.
\end{corollary}
\begin{proof}
For the hyperbolic quadrics, $t_r$ with $Q_0(r)=0$ carries $Q_r$ to $Q_0$.
For the elliptic ones, $\Sp(8,2)$ is transitive on the quadratic forms of a given type with polar form $B$~\cite{Taylor92}.
\end{proof}

The $136$ hyperbolic quadrics fall into three classes by their number of negative lines ($532$, $604$ and $612$)~\cite{Muller23}; re-signing (\cref{lem:resign}) and Clifford transport (\cref{lem:clifford}) change this number without changing the degree, so it plays no role below.

\section{Grids, doilies and their degree}\label{sec:doily}

\begin{lemma}[Odd sets]\label{lem:odd}
Let $\mathcal F$ be a family of contexts in which every point occurs an even number of times and whose signs have odd sum. Then every assignment violates an odd number, and in particular at least one, of the contexts of $\mathcal F$.
\end{lemma}
\begin{proof}
Modulo $2$, the number of violated contexts is $\sum_{C\in\mathcal F}\bigl(\sum_{p\in C}x(p)+s(C)\bigr)=\sum_C\sum_{p\in C}x(p)+\sum_Cs(C)=0+1$.
\end{proof}

A \emph{grid} is a generalised quadrangle $GQ(2,1)$: nine points $p_{ij}$ ($i,j\in\{1,2,3\}$) and six lines, the three rows $\{p_{i1},p_{i2},p_{i3}\}$ and the three columns $\{p_{1j},p_{2j},p_{3j}\}$. The rows form one \emph{ruling} and the columns the other.

\begin{lemma}[Peres--Mermin parity]\label{lem:mermin}
Let nine Pauli observables $P_{ij}$ form a grid in which every row and every column is a context, and in which $P_{ij}$ and $P_{kl}$ anticommute whenever $i\ne k$ and $j\ne l$. Then the six contexts have signs of odd sum.
\end{lemma}
\begin{proof}
Let $r_i$ and $c_j$ be the signs of the rows and the columns, so that the products of the rows and the columns are $(-1)^{r_i}I$ and $(-1)^{c_j}I$.
The product $M=P_{11}P_{12}P_{13}P_{21}\cdots P_{33}$ in row-major order equals $(-1)^{r_1+r_2+r_3}I$.
Reordering it into column-major order $P_{11}P_{21}P_{31}P_{12}\cdots P_{33}$ exchanges the relative order of $P_{ij}$ and $P_{kl}$ exactly when $i<k$ and $j>l$.
There are nine such pairs and each anticommutes, so the column-major product is $(-1)^9M$.
It also equals $(-1)^{c_1+c_2+c_3}I$, so $\sum r_i+\sum c_j$ is odd.
\end{proof}

By \cref{lem:odd} (each point lies on one row and one column), every assignment violates an odd number of the six lines of such a grid.
This is the familiar argument for the Peres--Mermin square~\cite{Peres90,Mermin90}; it also shows that every grid of $W(2N-1,2)$ is contextual~\cite{Holweck17}.

A \emph{doily} of a quadric $Q$ of $W(2N-1,2)$ is the configuration of singular points and lines of $Q$ inside a $5$-dimensional subspace $U\le V$ on which $Q$ is nondegenerate. Since $\dim U$ is odd and the characteristic is $2$, the restriction of $B$ to $U$ is always degenerate; nondegenerate means that its radical $U\cap U^\perp$ is one-dimensional, spanned by a vector $r$ with $Q(r)=1$, so that $Q$ is a parabolic form on $U$.
It is a parabolic quadric $\mathcal Q(4,2)$, a generalised quadrangle $GQ(2,2)$ with $15$ points and $15$ lines~\cite{PayneThas09,HirschfeldThas16}.

\begin{lemma}\label{lem:collinear}
Let $D$ be a doily of a quadric and $G\subseteq D$ a grid (nine points and six lines of $D$ forming a $GQ(2,1)$). Two points of $G$ that are not collinear in $G$ anticommute.
\end{lemma}
\begin{proof}
Let $p,q\in G$ be non-collinear in $G$.
Two points of $D$ commute if and only if they are collinear in $D$, because $p+q\in U$ and $B(p,q)=0$ means that $\{p,q,p+q\}$ is a line of the quadric.
Suppose $p,q$ are collinear in $D$.
Let $M$ be a line of $G$ through $q$; then $p\notin M$.
In the grid, $p$ is collinear with exactly one point $m\ne q$ of $M$.
So in $D$ the point $p$ is collinear with two points $q,m$ of the line $M$, contradicting the axiom of a generalised quadrangle.
\end{proof}

\begin{lemma}\label{lem:sylvester}
A doily contains ten grids such that every line of the doily lies in exactly four of them and any two intersecting lines lie in exactly two of them.
\end{lemma}
\begin{proof}
$GQ(2,2)$ is unique up to isomorphism and has Sylvester's model~\cite{PayneThas09} (\cref{fig:doily}): the points are the $15$ two-element subsets (duads) of $\{1,\dots,6\}$ and the lines are the $15$ partitions of $\{1,\dots,6\}$ into three duads (synthemes), with inclusion as incidence.
For a partition $\{A,B\}$ of $\{1,\dots,6\}$ into two triples, the $9$ duads meeting both $A$ and $B$ correspond to $A\times B$, and the synthemes made of such duads to the $6$ bijections $A\to B$.
Two bijections $\pi,\rho$ share $|\{a:\pi(a)=\rho(a)\}|$ duads. Identifying $A$ and $B$ with $\{1,2,3\}$, this is $0$ when $\rho^{-1}\pi$ is a $3$-cycle and $1$ when it is a transposition.
So the even bijections are three pairwise disjoint lines covering the nine duads, and so are the odd ones, and every even bijection meets every odd one in exactly one duad: the nine duads and six synthemes form a grid.
This gives ten grids, one for each of the $\binom63/2=10$ partitions.
A syntheme $\{ab,cd,ef\}$ lies in the grid of $\{A,B\}$ exactly when each of its duads meets both $A$ and $B$, which leaves $2^3/2=4$ partitions.
Two synthemes sharing the duad $ab$, say $\{ab,cd,ef\}$ and $\{ab,ce,df\}$, both lie in the grid of $\{A,B\}$ iff $a,b$ are separated, $c$ is separated from $d$ and $e$, and $f$ from $d$ and $e$, which leaves the two partitions $\{\{a,c,f\},\{b,d,e\}\}$ and $\{\{b,c,f\},\{a,d,e\}\}$.
\end{proof}

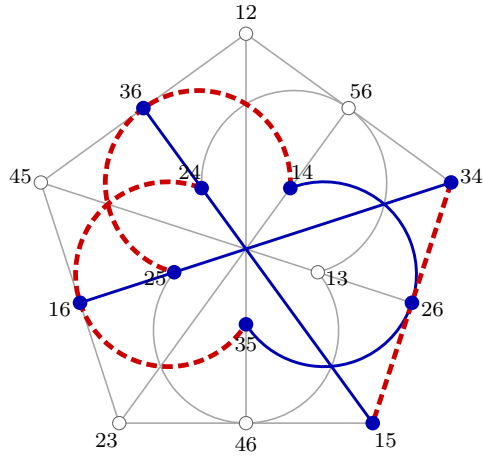
\begin{figure}[ht]
\centering
\begin{tikzpicture}[scale=0.95,
 line/plain/.style={black!35,line width=0.6pt},
 line/grid/.style={blue!70!black,line width=1.1pt},
 line/viol/.style={red!80!black,line width=1.8pt,dash pattern=on 4pt off 2pt},
 pt/grid/.style={circle,fill=blue!70!black,inner sep=0pt,minimum size=5.5pt},
 pt/plain/.style={circle,fill=white,draw=black!60,inner sep=0pt,minimum size=5pt}]
\draw[line/plain] (2.853,0.927) -- (0.000,3.000);
\draw[line/plain] (0.999,-0.324) arc[start angle=39.19, end angle=-219.19, radius=1.288];
\draw[line/plain] (1.763,-2.427) -- (-1.763,-2.427);
\draw[line/plain] (1.427,1.964) -- (-1.763,-2.427);
\draw[line/plain] (-1.763,-2.427) -- (-2.853,0.927);
\draw[line/plain] (-0.617,0.849) arc[start angle=-176.81, end angle=-435.19, radius=1.288];
\draw[line/plain] (-2.853,0.927) -- (0.000,3.000);
\draw[line/plain] (0.000,-2.427) -- (0.000,3.000);
\draw[line/plain] (2.308,-0.750) -- (-2.853,0.927);
\draw[line/grid] (0.617,0.849) arc[start angle=111.19, end angle=-147.19, radius=1.288];
\draw[line/grid] (-1.427,1.964) -- (1.763,-2.427);
\draw[line/grid] (-2.308,-0.750) -- (2.853,0.927);
\draw[line/viol] (0.617,0.849) arc[start angle=-3.19, end angle=255.19, radius=1.288];
\draw[line/viol] (2.853,0.927) -- (1.763,-2.427);
\draw[line/viol] (-0.617,0.849) arc[start angle=68.81, end angle=327.19, radius=1.288];
\node[pt/plain] at (0.000,3.000) {};
\node[font=\scriptsize] at (0.000,3.300) {12};
\node[pt/grid] at (2.853,0.927) {};
\node[font=\scriptsize] at (3.138,1.020) {34};
\node[pt/grid] at (1.763,-2.427) {};
\node[font=\scriptsize] at (1.940,-2.670) {15};
\node[pt/plain] at (-1.763,-2.427) {};
\node[font=\scriptsize] at (-1.940,-2.670) {23};
\node[pt/plain] at (-2.853,0.927) {};
\node[font=\scriptsize] at (-3.138,1.020) {45};
\node[pt/plain] at (1.427,1.964) {};
\node[font=\scriptsize] at (1.603,2.206) {56};
\node[pt/grid] at (0.617,0.849) {};
\node[font=\scriptsize] at (0.782,1.076) {14};
\node[pt/grid] at (2.308,-0.750) {};
\node[font=\scriptsize] at (2.594,-0.843) {26};
\node[pt/plain] at (0.999,-0.324) {};
\node[font=\scriptsize] at (1.265,-0.411) {13};
\node[pt/plain] at (0.000,-2.427) {};
\node[font=\scriptsize] at (0.000,-2.727) {46};
\node[pt/grid] at (0.000,-1.050) {};
\node[font=\scriptsize] at (0.000,-1.330) {35};
\node[pt/grid] at (-2.308,-0.750) {};
\node[font=\scriptsize] at (-2.594,-0.843) {16};
\node[pt/grid] at (-0.999,-0.324) {};
\node[font=\scriptsize] at (-1.265,-0.411) {25};
\node[pt/grid] at (-1.427,1.964) {};
\node[font=\scriptsize] at (-1.603,2.206) {36};
\node[pt/grid] at (-0.617,0.849) {};
\node[font=\scriptsize] at (-0.782,1.076) {24};
\end{tikzpicture}
\caption{A doily in Sylvester's model: the points are the duads $ij$ of $\{1,\dots,6\}$ and the lines the synthemes (\cref{lem:sylvester}). The filled points and the coloured lines form the grid of the partition $\{1,2,3\}\mid\{4,5,6\}$; the three dashed lines are one of its rulings. By \cref{prop:ruling}, an optimal assignment violates exactly such a ruling in every doily of the quadric.}
\label{fig:doily}
\end{figure}

\begin{remark}
These ten are in fact all the grids of the doily, but only the stated multiplicities are used below. We also checked, on all $40\,320$ doilies of the four-qubit hyperbolic quadric, that each has exactly ten grids with these multiplicities (\cref{app:computations}).
\end{remark}

\begin{proposition}\label{prop:doily}
Every doily of a quadric of $W(2N-1,2)$, with the signs inherited from the Pauli operators, has contextuality degree at least $3$. If an assignment violates exactly three of its lines, then they are pairwise disjoint and form a ruling of one of its grids, and each of the other nine grids contains exactly one of them.
\end{proposition}
\begin{proof}
Let $x$ be an assignment and $V_D$ the set of violated lines of the doily $D$.
Let $G$ run over the ten grids of \cref{lem:sylvester}.
By \cref{lem:collinear,lem:mermin,lem:odd}, every $G$ meets $V_D$ in an odd number of lines, and by \cref{lem:sylvester}, $\sum_G|V_D\cap G|=4|V_D|$.
The left side is a sum of ten odd numbers, so $4|V_D|\ge10$ and $|V_D|\ge3$.
If $|V_D|=3$, the ten odd numbers sum to $12$, so exactly one grid $G^*$ contains three violated lines and the other nine contain one each.
Two intersecting violated lines would lie in two common grids (\cref{lem:sylvester}), each then containing at least two, hence three, violated lines; so the three violated lines are pairwise disjoint.
They lie in $G^*$, and three pairwise disjoint lines of a grid form a ruling.
\end{proof}

The bound is attained: every doily has degree exactly $3$~\cite[Prop.~1]{Muller22doilies}, \cite{deBoutray22}, which we also checked by brute force over the $2^{15}$ assignments of each doily of the four-qubit hyperbolic quadric.

\section{The main theorem}\label{sec:main}

\begin{theorem}\label{thm:main}
Every hyperbolic quadric $\Qp(7,2)$ of the four-qubit space, with its $1575$ lines as contexts, has contextuality degree exactly $315$.
\end{theorem}

\begin{proof}
\emph{Lower bound.}
Let $Q$ be a hyperbolic quadric and $\mathcal D$ the set of its doilies.
The orthogonal group $\GO(Q)$ maps doilies to doilies, and it acts transitively on the lines of $Q$ by Witt's theorem~\cite{Taylor92}, since any linear isometry between two totally singular $2$-spaces extends to $V$.
The set $\mathcal D$ is not empty (an example is given in \cref{app:witness}); so every line lies in the same number $r\ge1$ of doilies, and counting incidences gives $15|\mathcal D|=1575\,r$.
For an assignment $x$ with violated set $\mathcal V$, \cref{prop:doily} gives
\[r\,|\mathcal V|=\sum_{D\in\mathcal D}|\mathcal V\cap D|\ \ge\ 3|\mathcal D|,\qquad\text{so}\qquad|\mathcal V|\ \ge\ \frac{3|\mathcal D|}{r}=\frac{3\cdot1575}{15}=315.\]
(In fact $|\mathcal D|=40\,320$ and $r=384$.)

\emph{Upper bound.}
The assignment of \cref{app:witness} violates exactly $315$ of the $1575$ lines of the real quadric $Q_0$.
Its existence was first shown by Muller \emph{et al.}~\cite{Muller24}; ours was found independently by local search.
By \cref{cor:allquadrics}, every hyperbolic quadric has an assignment with $315$ violations; explicitly, for $Q_0(r)=0$ the assignment $y(v)=x(t_rv)+\tau(v)$ of \cref{lem:clifford} works for $Q_r$, which we also checked for all $135$ such $r$.
\end{proof}

\begin{remark}[The role of the doily]
Averaging with grids instead of doilies gives only $|\mathcal V|\ge1575/6$, that is $\dg\ge263$, because a grid can have a single violated line out of six.
The doily gains from the rounding $\lceil10/4\rceil=3$.
Its violated fraction $3/15=1/5$ is exactly the violated fraction $315/1575$ of the quadric, so the lower bound is tight.
\end{remark}

\begin{remark}[Averaging over three-qubit sections]
The quadric $Q_0$ has $120$ parabolic hyperplane sections $\mathcal Q(6,2)$, the sections by $a^\perp$ with $Q_0(a)=1$; each is a copy of $W(5,2)$ with $315$ lines, and each line of $Q_0$ lies in $24$ of them.
If each section, with its inherited signs, has degree $63$, which follows from~\cite{Saniga24hex} together with~\cite[Thm.~5]{MullerGiorgetti24}, the same averaging gives $\dg\ge120\cdot63/24=315$.
This route was implicitly available from~\cite[\S5.1.2]{Muller24}, which notes the $120$ sections, but it was not drawn there. Since the degree $63$ of $W(5,2)$ is itself obtained by averaging over doilies~\cite[\S4.1]{Saniga24hex}, it is the doily argument in two stages; the direct argument needs only \cref{prop:doily}.
\end{remark}

\section{The optimal assignments}\label{sec:optimal}

Equality in the averaging argument constrains every optimal assignment locally.

\begin{proposition}\label{prop:ruling}
Let $x$ be an assignment of a hyperbolic quadric $\Qp(7,2)$ violating exactly $315$ lines. Then every one of its $40\,320$ doilies contains exactly three violated lines; they are pairwise disjoint and form a ruling of one of the ten grids of the doily, and each of the other nine grids contains exactly one of them.
Conversely, if $\mathcal F$ is a family of doilies covering every line equally often, an assignment violating exactly three lines of every doily of $\mathcal F$ is optimal.
\end{proposition}
\begin{proof}
If $|\mathcal V|=315$, the inequality $\sum_D|\mathcal V\cap D|\ge3|\mathcal D|$ is an equality, and each term is at least $3$, so each term is $3$. Apply \cref{prop:doily}.
Conversely, if $\mathcal F$ covers every line $r'$ times and every doily of $\mathcal F$ has three violated lines, then $r'|\mathcal V|=3|\mathcal F|$ and $15|\mathcal F|=1575\,r'$, so $|\mathcal V|=315$.
\end{proof}

The converse turns the global optimisation problem over $2^{135}$ assignments into local constraints, each on the fifteen lines of a doily, and this makes a complete enumeration easy.

\begin{theorem}[computer-assisted]\label{thm:dw}
For the real quadric $Q_0$ there are exactly $240$ sets of lines that are the violated set of an optimal assignment, and each is the violated set of exactly $256$ assignments, so that there are $61\,440$ optimal assignments.
The $240$ sets form a single orbit of the orthogonal group $\GO^+(8,2)$ acting on the lines.
Each of them, as a point--line geometry, is isomorphic to the dual polar space $DW(5,2)$: it has $315$ lines through all $135$ points, with $7$ lines through every point.
By \cref{lem:clifford} the same holds for every hyperbolic quadric of $W(7,2)$.
\end{theorem}

\begin{proof}[Proof (computer-assisted)]
\emph{Enumeration.} \Cref{sec:sat} gives a propositional formula whose models are exactly the optimal assignments (\cref{lem:encoding}). Enumerating its models with a SAT solver, each new violated set being excluded by a blocking clause, yields $240$ violated sets, after which the formula is unsatisfiable; that last unsatisfiability is certified by a DRAT proof checked by \texttt{drat-trim} and, converted to LRAT, by the verified checker \texttt{cake\_lpr}~\cite{cakeLPR}. The certificate covers only this unsatisfiability; the correctness of the encoding is \cref{lem:encoding}, proved by hand (see \cref{sec:sat}).
\emph{Gauge.}
Two assignments have the same violated set if and only if their difference $f$ satisfies $f(u)+f(v)+f(u+v)=0$ on every line.
The line--point incidence matrix has rank $127$ over $\F_2$, so these differences form a space of dimension $135-127=8$, namely the restrictions of the $256$ linear functionals on $V$.
Consequently, each realisable violated set comes from exactly $2^8=256$ assignments.
\emph{Orbit.}
The orthogonal group $\GO^+(8,2)$ of $Q_0$ is generated by the $120$ transvections $t_r$ with $Q_0(r)=1$ (reflections), each induced by a Clifford unitary~\cite{Taylor92}.
By \cref{lem:clifford} each maps violated sets of optimal assignments to violated sets of optimal assignments.
Starting from the violated set of our witness, the orbit under these generators has exactly $240$ elements; linear algebra over $\F_2$ confirms that each of them is the violated set of some assignment, independently of the SAT solver, and the orbit coincides with the enumerated list.
Since $|\GO^+(8,2)|=348\,364\,800$, the stabiliser of a violated set has order $1\,451\,520=|\Sp(6,2)|$, the order of the automorphism group of $DW(5,2)$.
\emph{Isomorphism.}
The dual polar space $DW(5,2)$ has as points the $135$ planes (maximal totally isotropic subspaces) of $W(5,2)$ and as lines its $315$ lines~\cite{DeBruyn06}.
For each of the $240$ violated sets, a backtracking search finds
a bijection from the planes of $W(5,2)$ to the points of $Q_0$.
We check that, for every line $\ell$ of $W(5,2)$, the three
planes containing $\ell$ map to a violated line, and that these
images are exactly the $315$ violated lines.
\end{proof}

\Cref{thm:dw} confirms the second half of the conjecture of~\cite{Muller24}: the unsatisfied configurations of the four-qubit hyperbolic quadric are \emph{always} isomorphic to $DW(5,2)$.
Muller \emph{et al.}\ explained why one should expect this: $DW(5,2)$ contains $120$ split Cayley hexagons of order two, as many as the quadric contains three-qubit subspaces $W(5,2)$, and the unsatisfied configurations of $W(5,2)$ are exactly the classically embedded hexagons~\cite{Muller24,Saniga24hex}.

\begin{remark}[Optimal assignments are not quadratic]
The assignments of the form $x(v)=q(v)$, for a quadratic form $q$ on $V$, violate the line $\{u,v,u+v\}$ exactly when $\beta_q(u,v)\ne s(\{u,v,u+v\})$, where $\beta_q$ is the polar form of $q$.
An exhaustive search over all $2^{28}$ alternating forms $\beta_q$ shows that none violates fewer than $532$ lines, the number of negative lines, attained by $x=0$.
So no optimal assignment is quadratic.
\end{remark}

\section{A lower bound for all quadrics}\label{sec:general}

The argument of \cref{sec:main} uses only the doilies of the quadric and the transitivity of its orthogonal group on lines.

\begin{theorem}\label{thm:general}
Let $N\ge3$ and let $\mathcal C$ be a hyperbolic quadric $\Qp(2N-1,2)$, an elliptic quadric $\Qm(2N-1,2)$, or the whole space $W(2N-1,2)$, with all its $L$ lines as contexts. Then $\dg\mathcal C\ge\lceil L/5\rceil$.
\end{theorem}
\begin{proof}
For a quadric $Q$, a nondegenerate $6$-dimensional subspace of type $+$ or $-$ exists for $N\ge3$ (take three of the hyperbolic or elliptic planes of an orthogonal decomposition), and its quadric $\mathcal Q^\pm(5,2)$ contains parabolic sections $\mathcal Q(4,2)$; so $Q$ has doilies.
$\GO(Q)$ is transitive on the lines of $Q$ by Witt's theorem, and \cref{prop:doily} applies to every doily.
The argument of \cref{thm:main} gives $|\mathcal V|\ge 3L/15$.
For $W(2N-1,2)$, use the doilies of all its quadrics: $\Sp(2N,2)$ permutes them and is transitive on the lines.
\end{proof}

For $W(2N-1,2)$ this recovers the bound $(4^N-1)(4^{N-1}-1)/15$ of~\cite[eq.~(11)]{Muller24}, which was obtained by a recursion over subspaces.
The number of lines follows from the fact that the lines of a quadric through a point $p$ correspond to the points of the quadric $p^\perp/p$ of rank one less: $L=n_N\,n_{N-1}/3$, with $n_N=(2^N-\varepsilon)(2^{N-1}+\varepsilon)$ points on $\mathcal Q^\varepsilon(2N-1,2)$.
\Cref{tab:general} compares the bound with the known exact values and with the heuristic upper bounds of~\cite{Muller24}.

\begin{table}[ht]
\centering\small
\begin{tabular}{@{}rlrrrr@{}}
\toprule
$N$ & configuration & points & lines $L$ & $\lceil L/5\rceil$ & known value or upper bound \\
\midrule
3 & $\Qp(5,2)$ & 35 & 105 & 21 & $=21$~\cite{deBoutray22} \\
3 & $\Qm(5,2)$ & 27 & 45 & 9 & $=9$~\cite{deBoutray22} \\
3 & $W(5,2)$ & 63 & 315 & 63 & $=63$~\cite{Muller23,Saniga24hex} \\
4 & $\Qp(7,2)$ & 135 & 1\,575 & 315 & $=315$ (\cref{thm:main}) \\
4 & $\Qm(7,2)$ & 119 & 1\,071 & 215 & $\le315$~\cite{Muller24} \\
4 & $W(7,2)$ & 255 & 5\,355 & 1\,071 & $\le1\,575$~\cite{Muller24} \\
5 & $\Qp(9,2)$ & 527 & 23\,715 & 4\,743 & $\le6\,975$~\cite{Muller24} \\
5 & $\Qm(9,2)$ & 495 & 19\,635 & 3\,927 & $\le7\,087$~\cite{Muller24} \\
6 & $\Qp(11,2)$ & 2\,079 & 365\,211 & 73\,043 & $\le132\,391$~\cite{Muller24} \\
6 & $\Qm(11,2)$ & 2\,015 & 332\,475 & 66\,495 & $\le131\,700$~\cite{Muller24} \\
7 & $\Qp(13,2)$ & 8\,255 & 5\,720\,715 & 1\,144\,143 & $\le2\,331\,191$~\cite{Muller24} \\
7 & $\Qm(13,2)$ & 8\,127 & 5\,458\,635 & 1\,091\,727 & $\le2\,294\,580$~\cite{Muller24} \\
\bottomrule
\end{tabular}
\caption{The doily bound of \cref{thm:general} against the known values. The quadric bounds for $N\ge5$ and for $\Qm(7,2)$ are, to our knowledge, the first lower bounds stated for these configurations.}
\label{tab:general}
\end{table}

The bound is tight in the four cases where the degree is known.
In the other cases the heuristic upper bounds are much larger, between $0.29L$ and $0.42L$.
Whether the doily bound is tight for $\Qp(9,2)$, as it is for $\Qp(5,2)$ and $\Qp(7,2)$, is the natural next question.

\section{Physical consequences}\label{sec:physics}

\paragraph{A noncontextuality inequality with an exact classical bound.}
For a configuration with contexts $C=\{u,v,u+v\}$, consider the sum of the signed correlators
\[\chi=\sum_C(-1)^{s(C)}\bigl\langle P(u)P(v)P(u+v)\bigr\rangle,\]
where each correlator is estimated by measuring the three commuting observables of $C$ jointly or in sequence~\cite{Cabello08,Cabello10}.
Quantum mechanics gives $\chi=N_C$, the number of contexts, in every state, because $(-1)^{s(C)}P(u)P(v)P(u+v)=I$.
A noncontextual model that assigns the outcomes $(-1)^{x(p)}$ satisfies the constraints of the non-violated contexts and violates the others, so $\chi\le N_C-2\dg\mathcal C$, and mixtures of such models do no better~\cite{Kelleher24}.

\begin{corollary}\label{cor:inequality}
For the $1575$ contexts of a four-qubit hyperbolic quadric, every noncontextual hidden-variable model satisfies $\chi\le945$, the bound is attained, and quantum mechanics gives $\chi=1575$ in every four-qubit state.
Equivalently, if a context is chosen uniformly at random,
the best strategy based on a context-independent assignment
of outcomes to observables succeeds with probability exactly
$4/5$, whereas ideal quantum measurements satisfy the
context constraint with certainty.
\end{corollary}

The inequality is exact, but it is not stronger than known ones: the ratio of the classical to the quantum value, $945/1575=3/5$, is the same as for the doily ($9/15$), the three-qubit hyperbolic quadric ($63/105$) and the generalised quadrangle $GQ(2,4)=\Qm(5,2)$ ($27/45$)~\cite{Kelleher24}.
Although $\Qp(7,2)$ has $105$ times as many contexts as the
doily, the two configurations have the same minimum violated
fraction, $1/5$.
If measurement noise uniformly scales every ideal signed
correlator by a visibility $\eta$, the inequality is violated
exactly when $\eta>3/5$, as for the doily.
Its interest is rather that the classical bound is known exactly for a configuration this large, and that for $Q_0$ all observables are real, so the test needs only real-amplitude (rebit) measurements.
By \cref{thm:general}, every quadric of $N\ge3$ qubits violates at least a fifth of its contexts in every noncontextual model, so its ratio is at most $3/5$; the open question is which of them have a violated fraction strictly above $1/5$, that is a ratio strictly below $3/5$.
The heuristic upper bounds of~\cite{Muller24} for $W(7,2)$, $\Qm(7,2)$ and $\Qp(9,2)$ are all exactly $5/17$ of their lines, which leaves open a ratio as low as $7/17$.

\paragraph{Where the contextuality sits.}
\Cref{prop:ruling} says that in every optimal noncontextual model, the failures are spread perfectly evenly: each doily of the quadric fails on exactly one ruling of one of its grids.
\Cref{thm:dw} adds that the failures form one rigid geometric object, a $DW(5,2)$, unique up to the symmetry of the quadric.
This is the four-qubit analogue of the three-qubit picture of~\cite{Saniga24hex}, where the unsatisfied lines of $W(5,2)$ always form a classically embedded split Cayley hexagon.

\section{Formal verification in Lean}\label{sec:lean}
We formalised the theorem for $Q_0$ in Lean~4 (v4.35.0-rc2) with Mathlib~\cite{Lean4,Mathlib}.
The toolchain is the release candidate v4.35.0-rc2 because it is the one required by the Mathlib commit we use (\texttt{045acef0}); both are pinned in the repository (\texttt{lean-toolchain}, \texttt{lake-manifest.json}), so \texttt{lake build} fetches exactly these versions, and the build was reproduced from a fresh checkout on a second machine.
The final statement is
\begin{center}
\texttt{theorem degree\_eq\_315 : Contextuality.Pauli.degree 0 = 315}
\end{center}
where \texttt{degree w} is the infimum, over all maps $x:\mathbb N\to\mathbb Z/2$, of the number of violated contexts of the quadric $Q_w$.
\texttt{\#print axioms degree\_eq\_315} reports only \texttt{propext}, \texttt{Classical.choice} and \texttt{Quot.sound}: the proof uses no \texttt{sorry}, no added axiom, and in particular no \texttt{native\_decide}, which would trust the compiler and its runtime.

\paragraph{Definitions.}
\texttt{Basic.lean} defines points as natural numbers $v<256$ with the encoding of \cref{sec:prelim} (bit $i$ is $x_i$, bit $i+4$ is $z_i$), the symplectic form, the quadratic forms $Q_w$, the phase exponent $\varphi$ as a sum of single-qubit exponents, the contexts as the pairs $(a,b)$ with $a<b<a\oplus b$ satisfying the definition (computed inside Lean by filtering all $256^2$ pairs, so that the list of contexts is not an input), their signs, violation, and the degree.
\texttt{PhaseMatrix.lean} proves that $\varphi$ is the exponent in the product of the actual $16\times16$ Pauli matrices over $\mathbb C$ (\texttt{pauliMat\_mul}), so that the sign of every context is the sign of the operator product (\texttt{context\_product}); this ties the combinatorial definition to the physics.

\paragraph{Proof.}
\texttt{OddSets.lean} proves \cref{lem:odd} (\texttt{exists\_violated\_of\_odd}) for contexts given with point masks.
The lower bound follows \cref{thm:main}, with a smaller family than all $40\,320$ doilies: the $2520$ doilies forming the union of the orbits of the qubit permutations $S_4$ on a selected set, which cover every line exactly $24$ times (found by integer programming over the orbits).
For each doily the data lists its fifteen lines in the order of a fixed labelling of $GQ(2,2)$, so that the ten grids of \cref{lem:sylvester} sit at the same fifteen positions in every doily.
The kernel then checks:
(i) the keys of the $2520\cdot15=37\,800$ listed lines form a permutation of the $1575$ contexts, each repeated $24$ times (\texttt{keys\_perm});
(ii) in every doily, each of the ten grids is an odd set: the exclusive-or of the point masks of its six lines is $0$ and their signs sum to $1$;
(iii) each of the fifteen positions lies in four grids (\texttt{sum\_gridPos}).
Together with \cref{lem:odd} this gives \texttt{three\_le\_doily} and $24|\mathcal V|\ge3\cdot2520$, hence \texttt{le\_numViolated}: $315\le|\mathcal V|$ for every assignment. A kernel evaluation of the witness (\texttt{witness\_check}) gives the upper bound.
Nothing produced by the scripts that generate the data is trusted: every property of the data used in the proof is re-evaluated by the kernel, and the data's sortedness is never assumed.

\paragraph{Making the kernel computation fit on a laptop.}
All finite checks are discharged by \texttt{decide +kernel}, which evaluates a decidable proposition in the kernel.
The kernel caches intermediate terms, so its memory grows with the number of reduction steps.
Three naive checks did not fit: evaluating \texttt{contexts 0} (a filter over $256^2$ pairs) in one check exceeded a $4$\,GB limit, and a single sort-and-compare of all $37\,800$ keys was estimated, from the measured growth, at about $10$\,GB.
We restructured the proof so that no check is large:
\begin{itemize}[leftmargin=*]
\item \texttt{Q8PlusContexts.lean} states the contexts as $64$ literal chunks and proves \texttt{contexts\_eq : contexts 0 = (List.range 64).flatMap cchunk} by $64$ separate kernel checks; every later check rewrites with it instead of recomputing the filter.
\item The permutation check (i) is split twice: for each of $26$ chunks of $100$ doilies, a merge sort of its $1500$ keys equals a literal bucketed by \texttt{key/1024} (\texttt{sortOK}); for each of the $64$ buckets, the merged pieces from the $26$ chunks equal the corresponding chunk of contexts, each repeated $24$ times (\texttt{bucketOK}).
The merge functions are proved to permute their input (\texttt{mergeTree\_perm}) and a lemma on exchanging two \texttt{flatMap}s up to permutation (\texttt{flatMap\_comm\_perm}) glues the chunks, so the correctness of the sorting is never assumed.
\item Each Lean file imports only the Mathlib modules it needs: importing all of Mathlib alone costs about $4$\,GB of memory.
\end{itemize}
The development has $191$ \texttt{decide +kernel} checks. \Cref{tab:lean} lists its files.

\begin{table}[ht]
\centering\small
\begin{tabular}{@{}lrrl@{}}
\toprule
File & lines & build & content \\
\midrule
\texttt{Basic.lean} & 133 & 10\,s & definitions; \texttt{degree} \\
\texttt{OddSets.lean} & 144 & 34\,s & \cref{lem:odd}; merge sort and its permutation lemmas \\
\texttt{PhaseMatrix.lean} & 86 & 27\,s & phase rule for the Pauli matrices over $\mathbb C$ \\
\texttt{Q8Plus.lean} & 240 & 16\,s & the proof of \texttt{degree\_eq\_315} \\
\texttt{Q8PlusData.lean} (generated) & 6\,266 & 55\,s & $2520$ doilies, grid positions, witness \\
\texttt{Q8PlusContexts.lean} (generated) & 580 & 71\,s & contexts in $64$ chunks; $67$ kernel checks \\
\texttt{Q8PlusChecks.lean} (generated) & 226 & 276\,s & $120$ kernel checks: cover, sorting, grids \\
\bottomrule
\end{tabular}
\caption{The Lean development. Build times with one thread on the machine of \cref{sec:costs}, Mathlib precompiled; the whole build takes under ten minutes, and the peak memory, reached on \texttt{Q8PlusChecks.lean}, is 2.2\,GB.}
\label{tab:lean}
\end{table}

The formalisation covers the real quadric $Q_0$. The other $135$ hyperbolic quadrics follow from \cref{cor:allquadrics}, which we checked by computer but did not formalise.
The Lean proof does not formalise the general arguments of \cref{sec:doily}: instead of \cref{lem:collinear,lem:sylvester}, the kernel checks the grid structure and the odd-set property directly on each of the $2520$ doilies it uses. Likewise \cref{thm:general} for other quadrics and other $N$ is not formalised.
\Cref{thm:dw} is not formalised; its certificate is the DRAT proof of \cref{sec:sat}, checked by \texttt{drat-trim} and by \texttt{cake\_lpr}.

\section{Computations: encodings, certificates and costs}\label{sec:computations}
This section is written for readers who want to check or reuse the computations. All of them are in the repository; \cref{tab:costs} lists their running times.

\subsection{What is computed, and who checks it}
\begin{itemize}[leftmargin=*]
\item \emph{Finite facts used in the proofs} (numbers of quadrics, points, lines and signs; the doilies, their grids and multiplicities; the witness): computed by one Python program, recomputed by a second one written separately, and, for $Q_0$, re-evaluated by the Lean kernel. The signs are computed both from the phase rule and from $16\times16$ matrix products.
\item \emph{\Cref{thm:main}}: the lower bound is a proof by hand (\cref{sec:doily,sec:main}); the upper bound is the evaluation of the witness on $1575$ lines. Both are kernel-checked in Lean for $Q_0$.
\item \emph{\Cref{thm:dw}}: the completeness of the list of $240$ violated sets rests on the SAT encoding of \cref{sec:sat} and on a DRAT proof checked by \texttt{drat-trim} and, as LRAT, by the verified checker \texttt{cake\_lpr}. The orbit, the realisability of each set and the isomorphisms with $DW(5,2)$ are checked by separate code that does not use the SAT solver.
\item \emph{Remark on quadratic assignments}: an exhaustive search in C.
\end{itemize}

\subsection{The SAT encoding of the optimal assignments}\label{sec:sat}
The degree problem is an instance of MAX-XOR-SAT: $135$ Boolean variables and $1575$ parity constraints of length three, one per line, and the task is to satisfy as many as possible.
Rather than optimise, we use \cref{prop:ruling} to characterise the optimum by local constraints.
\begin{itemize}[leftmargin=*]
\item \emph{Variables.} $x_p$ for the $135$ points and $v_C$ for the $1575$ lines.
\item \emph{Parity.} For each line $C=\{a,b,c\}$, the eight clauses that forbid the assignments of $(v_C,x_a,x_b,x_c)$ with $v_C+x_a+x_b+x_c\neq s(C)$, so that $v_C$ holds iff $x$ violates $C$ ($12\,600$ clauses).
\item \emph{Local cardinality.} For each doily $D$ of a family $\mathcal F$ covering every line equally often, ``exactly three of the fifteen $v_C$, $C\in D$, are true'', encoded by PySAT's \texttt{CardEnc.equals} with its sequential-counter encoding~\cite{Sinz05}, which combines an ``at most $3$'' and an ``at least $3$'' counter ($144$ clauses and $72$ auxiliary variables per doily; these are PySAT's sizes, smaller than those of the textbook encoding). We use either the $2520$ doilies of the Lean proof or all $40\,320$.
\item \emph{Pair clauses} (redundant). For any two intersecting lines $C,C'$ of a common doily, the clause $\neg v_C\vee\neg v_{C'}$. There are $37\,800$ such pairs: through each point, two of its $35$ lines lie in a common doily iff they are not coplanar on the quadric. These clauses are not derived by the solver: they are justified by \cref{prop:doily}, applied to every doily of the quadric (not only to those of $\mathcal F$), and they are part of the input formula whose unsatisfiability the certificate checks.
\item \emph{Blocking clauses.} After a model with violated set $\mathcal V$, the clause $\bigvee_{C\in\mathcal V}\neg v_C$. Since every model has $|\mathcal V|=315$, it excludes exactly the models with violated set $\mathcal V$.
\end{itemize}

\begin{lemma}\label{lem:encoding}
With either family, with or without the pair clauses, the projections of the models onto the $x$-variables are exactly the optimal assignments of $Q_0$.
\end{lemma}
\begin{proof}
By \cref{prop:ruling} (converse), every model is optimal. Conversely, an optimal assignment has exactly three violated lines in every doily, which are pairwise disjoint by \cref{prop:doily}, so it satisfies all cardinality and pair clauses with the auxiliary variables of the sequential counters set accordingly.
\end{proof}

No symmetry breaking is used: the gauge freedom of \cref{sec:optimal} leaves $256$ models per violated set, and the blocking clauses act on the $v$-variables only.
The final formula, the base formula plus the $240$ blocking clauses, has $183\,150$ variables and $413\,520$ clauses in the default configuration.
CaDiCaL~3.0.1~\cite{CaDiCaL} refutes it and writes a DRAT proof of $39$\,MB; \texttt{drat-trim}~\cite{DRATtrim} verifies it, using $120\,460$ of the $560\,025$ lemmas in the run recorded in the repository (CaDiCaL's proof differs slightly between machines; the formula itself is fixed by its SHA-256 hash in \texttt{drat/check.log}).
The proof was also converted to LRAT by \texttt{drat-trim} ($22$\,MB) and checked by the verified checker \texttt{cake\_lpr}~\cite{cakeLPR}, whose correctness is proved in HOL4 down to machine code.
The certificate establishes only that the final formula is unsatisfiable. Three parts of the argument are trusted rather than certified: the encoder, a short Python program that writes the formula; \cref{lem:encoding}, including the justification of the pair clauses by \cref{prop:doily}; and the enumeration loop, whose only role is to add one blocking clause per violated set it finds, all of which appear in the certified final formula.

\subsection{What made the enumeration easy}\label{sec:ablation}
Muller \emph{et al.}~\cite{Muller23,Muller24} attacked the degree with SAT solvers and with heuristics directly, and the direct formulations remain hard.
\Cref{tab:ablation} compares, with the same solver and a time limit of $30$ minutes, two direct formulations with the four variants of our encoding.
The direct formulations are a MaxSAT instance (the parity clauses as hard clauses, $\neg v_C$ as $1575$ soft clauses, solved by RC2~\cite{RC2}) and the decision problem ``at most $314$ violated lines'' with a totalizer encoding~\cite{Totalizer}, whose unsatisfiability is the lower bound of \cref{thm:main}.

\begin{table}[ht]
\centering\small
\resizebox{\linewidth}{!}{%
\begin{tabular}{@{}llrrl@{}}
\toprule
Formulation & & variables & clauses & result within 30 min \\
\midrule
MaxSAT (RC2) & & 1\,710 & 12\,600 + 1\,575 soft & not solved \\
SAT, at most $314$ violated (totalizer) & & 18\,562 & 1\,268\,978 & not solved \\
\midrule
local cardinality, $2520$ doilies & & 183\,150 & 375\,480 & all 240 by 44\,s; not refuted \\
local cardinality, $40\,320$ doilies & & 2\,904\,750 & 5\,818\,680 & all 240 by 1363\,s; refuted at 1447\,s \\
local cardinality, $2520$ doilies + pairs & & 183\,150 & 413\,280 & all 240 by 29\,s; refuted at 31\,s \\
local cardinality, $40\,320$ doilies + pairs & & 2\,904\,750 & 5\,856\,480 & all 240 by 403\,s; refuted at 407\,s \\
\bottomrule
\end{tabular}}
\caption{Formulations of the degree problem for $Q_0$, solved by CaDiCaL~1.5.3 through PySAT (RC2 with its default Glucose~3 oracle) on the machine of \cref{sec:costs}, one thread, $30$-minute limit. For the enumerations, the times are until the last of the $240$ violated sets is found and until the final formula is refuted.}
\label{tab:ablation}
\end{table}

Three observations stand out.
First, neither direct formulation finishes: RC2 does not reach the optimum and the solver does not refute ``at most $314$'' within $30$ minutes, although \cref{thm:main} shows that the answer is $315$.
The lower bound is therefore not merely faster to prove by hand; for these generic encodings it is out of reach of the solver in this time.
Second, the local cardinality constraints make the \emph{search} easy: every variant finds all $240$ optimal violated sets.
Third, \emph{completeness}, the refutation after the last blocking clause, depends on the redundant pair clauses: without them the small family is not refuted within $30$ minutes and the large one needs $84$\,s for the last call alone; with them the final call takes $2$\,s and $4$\,s.
The pair clauses encode the second half of \cref{prop:doily}, that the violated lines of a doily are pairwise disjoint, so they are a lemma of the proof handed to the solver.
The smaller family is also better: the $2520$ doilies of the Lean proof give a formula $14$ times smaller than all $40\,320$, and the full enumeration takes $31$\,s instead of $407$\,s.

\subsection{Running times}\label{sec:costs}
All times in \cref{tab:costs} are wall-clock times on a laptop (Apple M3, $8$ cores, $16$\,GB, macOS~14.5), with one thread and at low priority while other programs were running, so they are upper bounds on what a dedicated machine needs.
Software: Python~3.12.12 with NumPy~2.5.4 and PySAT~1.9 (which bundles CaDiCaL~1.5.3); the standalone CaDiCaL~3.0.1; \texttt{drat-trim} at commit \texttt{2e3b2dc}; \texttt{cake\_lpr} at commit \texttt{a36874a}; Lean~4 v4.35.0-rc2 with Mathlib at commit \texttt{045acef0}.

\begin{table}[ht]
\centering\small
\resizebox{\linewidth}{!}{%
\begin{tabular}{@{}lll@{}}
\toprule
Computation & tool & time \\
\midrule
all finite facts of the paper (\texttt{verify\_facts.py}) & Python & 47\,s \\
\quad with the degree $3$ of all $40\,320$ doilies by brute force over $2^{15}$ assignments & Python, NumPy & 55\,s \\
separately written re-check of the bound and the witness (\texttt{independent\_check.py}) & Python, NumPy & 72\,s \\
search for a witness (\texttt{find\_witness.py}) & Python & under 1\,s \\
enumeration of the $240$ optimal violated sets (\texttt{enumerate\_optimal.py}) & CaDiCaL~1.5.3 & 30\,s \\
\quad of which the final refutation & & 2\,s \\
DRAT proof of the final refutation ($39$\,MB) & CaDiCaL~3.0.1 & 25\,s \\
DRAT check & \texttt{drat-trim} & 11--17\,s \\
LRAT check of the converted proof & \texttt{cake\_lpr} & 3\,s \\
orbit, realisability, $240$ isomorphisms with $DW(5,2)$ (\texttt{optimal\_structure.py}) & Python & 4\,s \\
table of lower bounds for all quadrics (\texttt{general\_bounds.py}) & Python & 0.1\,s \\
all $2^{28}$ quadratic assignments & C & 2.4\,s \\
Lean build of \texttt{degree\_eq\_315} (Mathlib precompiled) & Lean~4 & under 10\,min \\
axiom report & Lean~4 & 15\,s \\
\bottomrule
\end{tabular}}
\caption{Running times of all computations (wall clock, one thread).}
\label{tab:costs}
\end{table}

\section{Open problems}\label{sec:open}
\begin{enumerate}[leftmargin=*]
\item \emph{The elliptic quadric and the whole space.} For four qubits, $215\le\dg\Qm(7,2)\le315$ and $1071\le\dg W(7,2)\le1575$ from \cref{thm:general} and~\cite{Muller24}. Every line of $W(7,2)$ lies on $24$ of the $120$ elliptic quadrics, so $\dg\Qm(7,2)=315$ would give $\dg W(7,2)\ge 120\cdot315/24=1575$ and hence equality. The doily bound cannot reach $315$ here, and neither can averaging over the $136$ parabolic sections $\mathcal Q(6,2)$ of $\Qm(7,2)$, which also gives only $\lceil136\cdot63/40\rceil=215$; larger sub-configurations are needed.
\item \emph{Five qubits.} Is $\dg\Qp(9,2)=4743$? The heuristic upper bound $6975$ of~\cite{Muller24} is far above.
\item \emph{A conceptual proof of \cref{thm:dw}.} \Cref{prop:ruling} characterises the optimal violated sets by local conditions on doilies. Is there a direct proof that these conditions force a $DW(5,2)$, for instance through the Lagrangian Grassmannian map from the planes of $W(5,2)$ to the points of $\Qp(7,2)$~\cite{Holweck14,Muller24}?
\item \emph{Higher rank.} Does a hyperbolic quadric $\Qp(2N-1,2)$ have degree $L/5$ whenever $5$ divides $L$, and are its optimal violated sets again a single orbit?
\end{enumerate}

\section*{Code and data availability}
All scripts, the witness, the SAT certificates and the Lean development are at \repo; \cref{app:computations} lists what each checks.

\appendix

\section{The witness and an example doily}\label{app:witness}
The following assignment of $Q_0$ violates exactly $315$ of its $1575$ lines. It gives the value $1$ (outcome $-1$) to the $71$ observables below and $0$ (outcome $+1$) to the other $64$ real four-qubit observables.
\begin{center}
\small\ttfamily
\begin{tabular}{@{}*{12}{l}@{}}
IIIX & IIIZ & IIXX & IIXZ & IIZX & IIZZ & IXIX & IXIZ & IXXZ & IXZX & IXZZ & IYIY \\
IYYI & IYYX & IYZY & IZII & IZIZ & IZXI & IZYY & IZZZ & XIII & XIIX & XIXX & XIYY \\
XIZI & XXII & XXXI & XXXX & XXYY & XXZI & XYXY & XYYZ & XZII & XZXI & XZXZ & XZZI \\
XZZX & YIIY & YIXY & YIYI & YIYZ & YXIY & YXXY & YXYI & YXZY & YYII & YYXI & YYZI \\
YZYX & YZZY & ZIIX & ZIIZ & ZIXX & ZIXZ & ZIZX & ZXIX & ZXIZ & ZXXX & ZXXZ & ZXYY \\
ZXZX & ZXZZ & ZYIY & ZYXY & ZYYI & ZYZY & ZZII & ZZXI & ZZYY & ZZZI & ZZZZ &
\end{tabular}
\end{center}
The $315$ violated lines pass through every point, seven through each, and form a $DW(5,2)$ (\cref{thm:dw}).

An example of a doily of $Q_0$: the singular points of the subspace spanned by \obs{XIII}, \obs{IXII}, \obs{ZIII}, \obs{IZIX} and \obs{YIYI}. Its fifteen points are
\obs{XIII}, \obs{IXII}, \obs{XXII}, \obs{ZIII}, \obs{ZXII}, \obs{IZIX}, \obs{XZIX}, \obs{ZZIX}, \obs{YYIX}, \obs{YIYI}, \obs{YXYI}, \obs{IYYX}, \obs{XYYX}, \obs{YZYX}, \obs{ZYYX},
and three of its fifteen lines are negative: $\{\obs{XXII},\obs{ZZIX},\obs{YYIX}\}$, $\{\obs{ZXII},\obs{XYYX},\obs{YZYX}\}$ and $\{\obs{XZIX},\obs{YXYI},\obs{ZYYX}\}$.

\section{Reproducing the computations}\label{app:computations}
From the folder \texttt{code} of the repository, with Python and the packages \texttt{numpy} and \texttt{python-sat}:
\begin{center}\small
\resizebox{\linewidth}{!}{%
\begin{tabular}{@{}ll@{}}
\toprule
command & reproduces \\
\midrule
\texttt{python verify\_facts.py -{}-brute} & every finite fact of \cref{sec:prelim,sec:doily,sec:main}; the witness \\
\texttt{python independent\_check.py} & the same, separately written \\
\texttt{python enumerate\_optimal.py -{}-cnf} & \cref{thm:dw}: the $240$ sets and the final formula \\
\texttt{sh drat/check.sh} & the DRAT proof, its check, and the LRAT check by \texttt{cake\_lpr} (needs CaDiCaL and \texttt{drat-trim}; \texttt{cake\_lpr} if present) \\
\texttt{python optimal\_structure.py} & \cref{thm:dw}: orbit, realisability, isomorphisms with $DW(5,2)$ \\
\texttt{sh sat\_ablation.sh}, \texttt{sh sat\_direct.sh} & \cref{tab:ablation} \\
\texttt{python general\_bounds.py} & \cref{tab:general} \\
\texttt{sh quadratic/run.sh} & the remark on quadratic assignments \\
\bottomrule
\end{tabular}}
\end{center}
In the folder \texttt{lean}, run \texttt{lake exe cache get} to download the compiled Mathlib, then \texttt{lake build}, and finally \texttt{lake env lean scripts/AxiomReport.lean} for the axiom report.
Every script exits with a non-zero status if a check fails.

\bibliographystyle{plain}
\bibliography{refs}
\end{document}